\documentclass{IEEEtran}
\usepackage{cite}
\usepackage{amsmath,amssymb,amsfonts}
\usepackage{graphicx}
\usepackage{textcomp}
\usepackage{algorithm}
\usepackage{algpseudocode}
\usepackage{cases}
\usepackage{color}
\usepackage{amsfonts}
\usepackage{indentfirst}
\usepackage{slashbox}
\usepackage{caption}
\usepackage{booktabs}
\usepackage{subfigure}
\usepackage{multirow}
\usepackage{amsthm}
\newtheorem{theorem}{Theorem}
\newtheorem{definition}{Definition}
\newtheorem{example}{Example}
\newtheorem{remark}{Remark}

\newtheorem{construction}{Construction}

\newcommand{\tabincell}[2]{\begin{tabular}{@{}#1@{}}#2\end{tabular}}
\begin{document}
\title{New coded caching schemes from placement delivery arrays}
\author{Mingming Zhang, Minquan Cheng, Jinyu Wang, Xi Zhong and Yishan Chen 
\thanks{M. Zhang, M. Cheng, J. Wan and X. Zhong are with Guangxi Key Lab of Multi-source Information Mining $\&$ Security, Guangxi Normal University,Guilin 541004, China (e-mail: ztw$\_$07@foxmail.com,chengqinshi@hotmail.com, mathwjy@163.com, Zhong19961225@outlook.com).}
\thanks{Y. Chen is with Business College, Guilin Tourism University, Guilin 541004, China (e-mail: glcys@163.com).}
}

\maketitle
\begin{abstract}
Coded caching schemes with low subpacketization and small transmission rate are desirable in practice due to the requirement of low implementation complexity and efficiency of the transmission. Placement delivery arrays (PDA in short) can be used to generate coded caching schemes. However, many known coded caching schemes have large memory ratios. In this paper, we realize that some schemes with low subpacketization generated by PDAs do not fully use the users' caching content to create multicasting opportunities and thus propose to overcome this drawback. As an application, we obtain two new schemes with low subpacketizations, which have significantly advantages on the memory ratio and transmission rate compared with the original scheme.
\end{abstract}

\begin{IEEEkeywords}
Coded caching scheme, placement delivery array, memory ratio, subpacketization.
\end{IEEEkeywords}

\maketitle

\section{Introduction}
\IEEEPARstart{W}{ireless}  networks  have  been imposed tremendous pressure on the data transmission during the peak traffic times due to the explosive increasing mobile services, especially the video streaming. A coded caching scheme, which has been recognized as efficient solutions to reduce this tremendous pressure, was proposed in \cite{MN} and has been rapidly used to in various settings such as D2D networks \cite{JCM}, hierarchical networks \cite{KNAD}, insecure channel \cite{STCD}, among others.
\subsection{Systems model}
In a centralized $(K,M,N)$ caching system, a single server containing $N$ files with the same length connects to $K$ users over a shared link and each user has a cache memory of size $M$ files. Denote the $N$ files by $\mathcal{W}=\{W_0,\ldots,W_{N-1}\}$ and $K$ users by $\mathcal{K}=\{0,\ldots,K-1\}$. An $F$-division $(K,M,N)$ centralized coded caching scheme consists of two separated phases as follows \cite{MN}:

\noindent{\bf Placement phase:} During the off peak traffic times, each file is subdivided into $F$ equal packets, \emph{i.e.}, $W_{i}=\{W_{i,j}:0\leq j< F\}$. Then each user caches some packets (or XOR packets) of the files from the server. If the packets of all files are cached directly in the placement phase,  then this is called uncoded placement; otherwise,  we call it coded placement. Denote $\mathcal{Z}_k$ the content cached by user $k$.

\noindent{\bf Delivery phase:} During the peak traffic times, each user randomly requests one file from the files set $\mathcal{W}$ independently. The request file number is denoted by $\mathbf{d}=(d_0,d_1,\cdots,d_{K-1})$, i.e., user $k$ requests the $d_k$-th file $W_{d_k}$, where $0\leq d_k<N$ and $k\in\mathcal{K}$. Once the server received the request $\mathbf{d}$, it broadcasts a coded signal (XOR of some required packets) of size at most $S_{{\bf d}}$ packets to users such that each user is able to recover its requested file with the help of its caching contents.

In this paper, we focus on the worst-case scenierio, i.e., all the users require different files. In this case,  the transmission rate of a coded caching scheme is defined as the maximal transmission amount among all the requests in delivery phase, i.e. $
R=\max_{{\bf d}\in[N]^K}\{R_{{\bf d}}\}$. Since the implementation complexity of a coded caching scheme increases along with its subpacketization level,  it is desirable to design a scheme with the transmission rate and the subpacketization as small as possible.

\subsection{Prior work}
We focus on the above centralized coded caching schemes when $K<N$. Maddah-Ali and Niesen \cite{MN} introduced
the first deterministic $F$-division $(K,M,N)$ coded caching scheme
with $F = {K \choose KM/N}$ when $KM/N$ is an integer. Obviously, the subpacketization $F={K\choose KM/N}$ increases so rapidly as $K$ increases, which makes  this scheme  inpractical when $K$ is large. It is well known that there exists a tradeoff between the transmission rate and the subpacketization level for the fixed number of users and the memory ratio. Indeed, there are many research papers focusing on constructing the coded caching schemes with lower subpacketization levels while sacrificing transmission rates a little bit, for instances,  \cite{CBK,K,SZG,STD,SDLT,TR,YCTC,YTCC} etc.

In \cite{SDLT}, Shanmugam et al. disovered that all the deterministic $F$-division coded caching schemes can be recasted  into an $F\times K$ combinatoric structure, which is called a placement delivery array (PDA). PDAs were introduced by Yan et al. \cite{YCTC} when $K\leq N$ when they showed that the Ali-Niesen scheme in \cite{MN} corresponds to a special class of PDA, which is referred to as MN PDA.  Indeed, it turns out that  PDAs are  good tools to construct coded caching schemes.   By means of PDA, Cheng et al. in \cite{CJYT} generalized the constructions of the PDAs in \cite{SZG}, \cite{YCTC} and \cite{CJTY}  and obtained  some schemes with more flexible memory size.
However, the subpacketization of these schemes increases exponentially with respect to the number of the users. Yan et al.  \cite{YTCC}  discovered  an equivalence between a PDA and a strong coloring in bipartite graph, and used the results of the strong coloring in bipartite graph proposed by Jennifer {\em et al.} in \cite{JA} to obtain a new class of PDA. From the results on optimality in \cite{JA}, one can check that the scheme in \cite{YTCC} has the smallest transmission rate among all the schemes with the same placement strategy. The revelant informtion of this scheme is listed in Table \ref{tab-known}. Clearly,  when $r$ is very small, the supbacketization of the scheme in Table \ref{tab-known} also increases exponentially with respect to  the user number. When $r$ is large, the memory ratio of this scheme approximates $1$. This fact severely limits its use in practice. There are also several others schemes with low subpacketization such as \cite{ASK,CBK,SDLT}, however,  most of them have the memory ratio close to $1$.

{\begin{table*}[http!]
\center
\caption{The scheme realized by the PDA in \cite{YTCC}\label{tab-known}}
\small{
\begin{tabular}{|c|c|c|c|c|c|}
\hline
Parameters & User Number $K$  & Caching ratio $\frac{M}{N}$
& Rate $R$   & Subpacketization $F$     \\ \hline
\tabincell{c}{$H,r,b,\lambda \in Z^{+}$, \\
$0 < r,b<H$,\\ $\lambda < \min{\{r,b\}}$,\\ $r+b-\lambda<H$}
&${H \choose r}$
& $1-\frac{{r \choose \lambda}{H-r \choose b-\lambda}}{{H \choose b}}$
& $\frac{{H \choose r+b-2\lambda}}{{H \choose b}}\min\{{H-(r+b-2\lambda)\choose \lambda},{r+b-2\lambda \choose r-\lambda}\}$
&${H \choose b}$\\ \hline
\end{tabular} }
\end{table*}}
\subsection{Contributions and organizations}
A coded caching scheme realized by a PDA has uncoded placement phase, which will be introduced in Section \ref{sec_prob}. We note that many packets cached by users  in some of these  coded caching schemes with low subpacketizations are not fully used in this phase, i.e., generating no multicasting opportunities due to the fact that some stars in the PDA are wasted. This fact will be introduced in detail in Remark \ref{remark-PDA-property} and Subsection \ref{subsec-new-scheme}. Then we  propose to adopt coded placement phase to reduce the number of packets and  the subpacketizations of the schemes realized by some well known PDAs, while keep the multicasting opportunity at each time slot unchanged. In particular, based on the well known PDAs obtained by strongly edge coloring in \cite{JA,YTCC}, we obtain two new schemes in Table \ref{tab-main}.
{\begin{table*}[http!]
\center
\caption{The new schemes with $K={H\choose r}$ for the same $H$, $r$, $b$ and $\lambda$ in Table \ref{tab-known}\label{tab-main}}
\small{
\setlength{\tabcolsep}{5.5mm}{
\begin{tabular}{|c|c|c|c|c|}
\hline
Reference &  Memory ratio $\frac{M}{N}$
& Rate $R$   & Subpacketization $F$   \\ \hline
Theorem \ref{th-main-1} & $1-\frac{{r\choose \lambda}{H-r\choose b-\lambda}}{{H\choose b}-\sum\limits_{i=0}^{\lambda-1}{r\choose i}{H-r\choose b-i}}$
& $\frac{{H\choose r+b-2\lambda}}{{H\choose b}-\sum\limits_{i=0}^{\lambda-1}{r\choose i}{H-r\choose b-i}}{H-(r+b-2\lambda)\choose \lambda}$
&${H\choose b}-\sum\limits_{i=0}^{\lambda-1}{r\choose i}{H-r\choose b-i}$\\[0.3cm]  \hline
Theorem \ref{th-main-2}&$1-\frac{{r\choose \lambda}{H-r\choose b-\lambda}}{{H\choose b}-\sum\limits_{i=\lambda+1}^{r}{r\choose i}{H-r\choose b-i}}$&$\frac{{H\choose r+b-2\lambda}}{{H\choose b}-\sum\limits_{i=\lambda+1}^{r}{r\choose i}{H-r\choose b-i}}{r+b-2\lambda\choose r-\lambda }$
&${H\choose b}-\sum\limits_{i=\lambda+1}^{r}{r\choose i}{H-r\choose b-i}$\\[0.3cm] \hline
\end{tabular} }}
\end{table*}}

From Table \ref{tab-main}, we can see that our new schemes have smaller memory ratio than the scheme in Table \ref{tab-known}. In Subsection \ref{sec-Panalyses}, we  show that our new schemes have smaller transmission rate than the scheme in Table \ref{tab-known}, under the same memory ratio. This is   due to the fact that all the packets of the requested files in our new schemes are fully used.

The rest of this paper is organized as follows. The relationship between a PDA and a coded caching scheme is explained  in Section \ref{sec_prob}. In Section \ref{sec-NSCs} we introduce our research motivation and present the main  method. As an application, we construct two new schemes in Section \ref{sec-strong}. Finally, we conclude the paper in Section \ref{conclusion}.

\section{Coded caching schemes realized by PDAs}\label{sec_prob}
In this paper, we use the following notations unless otherwise stated. We use bold capital letters and curlicue letters to denote arrays and sets respectively. For any positive integers $m$ and $t$ with $t< m$, let $[0,m)=\{0,1,\ldots,m-1\}$ and ${[0,m)\choose t}=\{\mathcal{T}\ |\   \mathcal{T}\subseteq [0,m), |\mathcal{T}|=t\}$, i.e., ${[0,m)\choose t}$ is the collection of all $t$-sized subsets of $[0,m)$.

Yan et al., in \cite{YCTC} proposed an interesting and simple combinatorial structure, called  a placement delivery array, which can be used to generate a coded caching scheme.
\begin{definition}(\cite{YCTC})
\label{def-PDA}
For  positive integers $K,F, Z$ and $S$, an $F\times K$ array  $\mathbf{P}=(p_{j,k})$, $0\leq j< F, 0\leq k< K$, composed of a specific symbol $``*"$  and $S$ integers in $[0,S)$, is called a $(K,F,Z,S)$ placement delivery array (PDA) if it satisfies the following conditions:

\noindent C$1$. The symbol $``*"$ appears $Z$ times in each column;

\noindent C2. Each integer in $[0,S)$ occurs at least once in the array;

\noindent C$3$. For any two distinct entries $p_{j_1,k_1}$ and $p_{j_2,k_2}$, $p_{j_1,k_1}=p_{j_2,k_2}=s$ is an integer only if $j_1\ne j_2$, $k_1\ne k_2$ and $p_{j_1,k_2}=p_{j_2,k_1}=*$.
\end{definition}
\begin{theorem}(\cite{YCTC})
\label{th-Fundamental}Using Algorithm \ref{alg:PDA}, an $F$-division $(K,M,N)$ coded caching scheme with $\frac{M}{N}=\frac{Z}{F}$ and transmission rate $R=\frac{S}{F}$ can be realized by a $(K,F,Z,S)$ PDA.
\end{theorem}

\begin{algorithm}[htb]
\caption{caching scheme based on PDA in \cite{YCTC}}\label{alg:PDA}
\begin{algorithmic}[1]
\Procedure {Placement}{$\mathbf{P}$, $\mathcal{W}$}
\State Split each file $W_i\in\mathcal{W}$ into $F$ packets, i.e., $W_{n}=\{W_{n,j}\ |\ j\in [0,F)\}$.
\For{$k\in \mathcal{K}$}
\State $\mathcal{Z}_k\leftarrow\{W_{n,j}\ |\ p_{j,k}=*, \forall~n\in [0,N)\}$
\EndFor
\EndProcedure
\Procedure{Delivery}{$\mathbf{P}, \mathcal{W},{\bf d}$}
\For{$s=0,1,\cdots,S-1$}
\State  Server sends $\bigoplus_{p_{j,k}=s,j\in [0,F), k\in[0,K)}W_{d_{k},j}$.
\EndFor
\EndProcedure
\end{algorithmic}
\end{algorithm}

\begin{example}
\label{E-pda}
It is easy to verify that the following array is a $(6,6,2,12)$ PDA.
\begin{eqnarray}\label{eq-exam-Alg}
\mathbf{P}=\left(\begin{array}{cccccccccc}
*	&	0	&	1	&	2	&	3	&	*	\\
0	&	*	&	4	&	5	&	*	&	6	\\
1	&	4	&	*	&	*	&	7	&	8	\\
2	&	5	&	*	&	*	&	9	&	10	\\
3	&	*	&	7	&	9	&	*	&	11	\\
*	&	6	&	8	&	10	&	11	&	*	
\end{array}\right)
\end{eqnarray}
Using Algorithm \ref{alg:PDA}, one can obtain a $6$-division $(6,2,6)$ coded caching scheme in the following way.

\noindent\textbf{Placement Phase}: rom Line 2 we have $W_n=\{W_{n,0},W_{n,1},W_{n,2},W_{n,3},W_{n,4},W_{n,5}\}$, $n\in [0,6)$. Then by Lines 3-5, the contents cached by users are $\mathcal{Z}_0=\mathcal{Z}_5=$ $\{W_{n,0},W_{n,5}\ |\ n\in[0,6)\}$, $\mathcal{Z}_1=\mathcal{Z}_4=\{W_{n,1},W_{n,4}\ |\ n\in[0,6)\}$ and $\mathcal{Z}_2=\mathcal{Z}_3=\{W_{n,2},W_{n,3}\ |\ n\in[0,6)\}$.

\noindent\textbf{Delivery Phase}: Assume that the request vector is $\mathbf{d}=(0,1,2,3,4,5)$. By Lines 8-10, the server sends the following coded signals at times slots $0-11$.
\begin{eqnarray*}\begin{array}{c|c|c} \hline
0: W_{0,1}\oplus W_{1,0}&  4: W_{1,2}\oplus W_{2,1}& 8: W_{2,5}\oplus W_{5,2}\\ \hline
1: W_{0,2}\oplus W_{2,0}&  5: W_{1,3}\oplus W_{3,1}& 9: W_{3,4}\oplus W_{4,3}\\ \hline
2: W_{0,3}\oplus W_{3,0}&  6: W_{1,5}\oplus W_{5,1}& 10: W_{3,5}\oplus W_{5,3}\\ \hline
3: W_{0,4}\oplus W_{4,0}&  7: W_{2,4}\oplus W_{4,2}& 11: W_{4,5}\oplus W_{5,4}\\ \hline
\end{array}
\end{eqnarray*}
\end{example}

From Algorithm \ref{alg:PDA} and Example \ref{E-pda}, the properties C1 and C2 imply that all the users have the same memory size and the server must send a coded signal at each time slot respectively. Furthermore we have the following observations.
\begin{remark}
\label{remark-PDA-property}In a $(K,F,Z,S)$ PDA $\mathbf{P}$, each column represents one user's caching contents, i.e., if $p_{j,k}=*$, then user $k$ has cached the $j$-th packet of all the files in the server. If $p_{j,k}=s$ is an integer, it means that the $j$-th packets of all the files are not stored by user $k$. Then the server finds out the corresponding rows that contain $s$, say $j$, and includes the required packets labeled by such $j$s to the delivery data at time slot $s$. The property C3 of the PDA guarantees that each user can get the requested packet. The occurrence number of $s$ is called the coded gain at time slot $s$. Clearly we prefer to design a scheme with the coded gain as large as possible at each time slot $s$.
\end{remark}
\section{New schemes from PDAs}
\label{sec-NSCs}
In this section, we  observe that there are some wasted packets cached by users in the coded caching schemes with low supbacketizations realized by some special PDAs and Algorithm \ref{alg:PDA}. Then we construct new optimized schemes obtained by these PDAs.
\subsection{Research motivations}
\label{subsec-new-scheme}
Let us consider the properties of a PDA again. Let $\mathbf{P}$ be a $(K,F,Z,S)$ PDA. For any integer $s\in [0,S)$, assume that the occurrence number of $s$ is $r_s$, say $p_{j_u,k_u}=s$, $0\leq u< r_s$, $0\le j_u< F$ and $0\le k_u<K$. Consider the subarray formed by rows $j_0,\cdots,j_{r_s-1}$ and columns $k_0,\cdots,k_{r_s-1}$, which is of order  $r_s\times r_s$ since $j_{u}\ne j_{v}$ and $k_u\ne k_v$ for all $0\le u\ne v< r_s$ from the definition of a PDA. Furthermore, we have $p_{j_u,k_v}=*$ for all $0\le u\ne v< r_s$. This subarray is equivalent to the following $r_s\times r_s$ array
\begin{eqnarray}\label{Eqn_Matrix_1_1}
\mathbf{P}^{(s)}=\left(\begin{array}{ccc}
      s &  \cdots & *\\
      \vdots  &\ddots & \vdots\\
      * & \cdots & s
    \end{array}\right)
\end{eqnarray}
with respect to row/column permutation. For a star entry $p_{j,k}=*$, we call it useful if it occurs in $\mathbf{P}^{(s)}$ for some integer $s\in[0,S)$, otherwise we call it useless. From Line 9 in Algorithm \ref{alg:PDA} and the subarray in \eqref{Eqn_Matrix_1_1}, the following statement holds. If $p_{j,k}=*$ is useful, the XOR of the requested packets (indicated by $s$) containing some $W_{n,j}$ with $n\neq d_k$ is transmitted by the server at time slot $s$, which generates the coded gain. If $p_{j,k}=*$ is useless, the $j$-th packet of all the files cached by user $k$ is never transmitted by the server, i.e., all the packets $W_{n,j}$ $(n\in[0,N))$ cached by user $k$ do not generate any coded gain.

Intuitively, the coded gain reduces when the subpacketization $F$ is reduced or  the number of useless stars in each row of a PDA is increased.  Take the PDA in \eqref{eq-exam-Alg} obtained by strongly edge coloring in \cite{JA,YTCC} for an example. It is easy to check that the stars at $p_{j,5-j}$ ($j\in[0,6)$) of $\mathbf{P}$ are useless. However, the authors showed that for the fixed edges, the number of different coloring is minimal in \cite{JA}. This implies that for the fixed placement strategy indicated by the star entries in \eqref{eq-exam-Alg}, the number of different integers in \eqref{eq-exam-Alg} is minimal, i.e., the transmission rate of the scheme realized by \eqref{eq-exam-Alg} and Algorithm \ref{alg:PDA} is minimal under the placement strategy indicated by \eqref{eq-exam-Alg}.

In fact, given a PDA, if there exist some useless stars, we can delete these useless stars and then further reduce the subpacketization and the memory ratio without reducing the coded gain at each time slot. Now let us see the array $\mathbf{P}$ in \eqref{eq-exam-Alg} again. First we delete all the useless stars and obtain the following array.
\begin{eqnarray}\label{eq-exam-Deleting}
\mathbf{P}'=\left(\begin{array}{c|c|c|c|c|c}
*	&	0	&	1	&	2	&	3	&		\\ \hline
0	&	*	&	4	&	5	&		&	6	\\ \hline
1	&	4	&	*	&		&	7	&	8	\\ \hline
2	&	5	&		&	*	&	9	&	10	\\ \hline
3	&		&	7	&	9	&	*	&	11	\\ \hline
	&	6	&	8	&	10	&	11	&	*	
\end{array}\right)
\end{eqnarray}

According to $\mathbf{P}'$ in \eqref{eq-exam-Deleting}, we can modify the placement phase as follows. Each file $W_n$, $n\in[0,6)$, is divided into $5$ packets, say $(W_{n,0},W_{n,1},\ldots,W_{n,4})$. Let $W_{n,5}=\sum^{4}_{j=0}W_{n,j}$. Using the caching strategy in Lines 3-5 in Algorithm \ref{alg:PDA}, each user $k$ caches $\mathcal{Z}_k=\{W_{n,j}\ |\ p'_{j,k}=*, \ j\in [0,6),n\in[0,6)\}$. For instance, user $0$ caches $\mathcal{Z}_0=\{W_{n,0}\ |\ n\in[0,6)\}$. Clearly the memory ratio of each user is $\frac{1}{5}$, which is smaller than the memory ratio $\frac{2}{6}$ in Example \ref{E-pda}. Then for any request vector ${\bf d}$, using the delivery phase in Algorithm \ref{alg:PDA}, each user can decode its requested file since each file $W_n$ can be recovered by any $5$ packets out of $\{W_{n,j}|j\in[0,6)\}$. Clearly, the coded gain at each time slot is the same as that of the original scheme in Example \ref{E-pda}.
\subsection{New schemes}
\label{sec-New scheme}
Followng the proposal  in Subsection \ref{subsec-new-scheme}, the following result can be obtained.
\begin{theorem}
\label{th-main}
For any $(K,F,Z,S)$ PDA $\mathbf{P}$, if there exist $Z'$ useless stars in each column, then we can obtain an $(F-Z')$-division $(K,M,N)$ coded caching scheme with $\frac{M}{N}=\frac{Z-Z'}{F-Z'}$ and transmission rate $R=\frac{S}{F-Z'}$, in which the coded gain at each time slot is the same as the original scheme realized by $\mathbf{P}$ and Algorithm \ref{alg:PDA}.
\end{theorem}
\begin{proof}
Assume that $\mathbf{P}$ is a $(K,F,Z,S)$ PDA where each column has $Z'$ useless stars. Deleting the $Z'$ useless stars in each column, we obtain a new array $\mathbf{P}'=(p'_{j,k})$, $j\in [0,F)$, $0\in [0,K)$. Clearly each column of $\mathbf{P}'$ has $Z'$ blanks, $Z-Z'$ stars and $F-Z$ integers.

Based on $\mathbf{P}'$, we modify the placement strategy in Algorithm \ref{alg:PDA} as follows: The server divides each file into $F-Z'$ equal-sized packets and then encodes them using an $(F,F-Z')$ maximum distance separable (MDS) code in an appropriate operation field \cite{Lint}. Denote the resulting encoded packets by $W_{n,0}$, $W_{n,1}$, $\ldots$, $W_{n,F-1}$ for each file $W_{n}$, $n\in[0,N)$. Using the caching strategy in Lines 3-5 in Algorithm \ref{alg:PDA}, each user $k$ caches $\mathcal{Z}_k=\{W_{n,j}\ |\  p'_{j,k}=*, j\in [0,F),n\in[0,N)\}$. Clearly the memory ratio of each user is $\frac{M}{N}=\frac{Z-Z'}{F-Z'}$.

In the delivery phase, we also use the delivery strategy in Algorithm \ref{alg:PDA} as follows: For any request vector ${\bf d}$, using Lines 7-11 of Algorithm \ref{alg:PDA}, each user can get exactly $F-Z'$ required coded packets by property C3 of Definition \ref{def-PDA} and Remark \ref{remark-PDA-property}. From the property of $(F,F-Z')$ MDS code, each user can recover its requested file. So the transmission rate is $R=\frac{S}{F-Z'}$. Furthermore, the coded gain at each time slot is the same as that of the original scheme realized by $\mathbf{P}$ since the occurrence number of each integer is unchanged.
\end{proof}

Given an appropriate PDA $\mathbf{P}$, using Theorem \ref{th-main}, we can obtain a new scheme which reduces the subpacketization and the memory ratio while keeping the coded gain at each time slot unchanged compared with the original scheme realized by $\mathbf{P}$. In the following we demonstrate  that there exist useless stars in some well known PDAs.
\section{New schemes based on the PDAs in \cite{YTCC}}
\label{sec-strong}
In this section, we take the well known PDAs constructed in \cite{JA,YTCC} to show the advantages of Theorem \ref{th-main}.

\subsection{Constructions from \cite{CLZW,JA,YTCC}}
The authors showed that the schemes based on the strong coloring in bipartite graph proposed by Jennifer {\em et al.} in \cite{JA} are equivalent to some special PDAs in \cite{YTCC}. For the readers' convenience, we  use the construction of such PDAs in \cite{CLZW}.
\begin{construction}(\cite{JA,YTCC,CLZW})
\label{constrct1}For any positive integers $H$, $b$, $r$, $\lambda$ satisfying $0<r,b<H$, $\lambda<\min\{r,b\}$ and $r+b\leq H+\lambda$ define $\mathcal{F}={[0,H)\choose b}$, $\mathcal{K}={[0,H)\choose r}$ and $\mathcal{I}={[0,H)\choose \lambda}$. Then we can obtain an $({H\choose r},{H\choose b},{H\choose b}-{r\choose \lambda}{H-r\choose b-\lambda},{H\choose r+b-2\lambda}{H-r-b+2\lambda\choose \lambda})$ PDA $\mathbf{P}=(p_{B,A})_{B\in \mathcal{F}, A\in \mathcal{K}}$ and an $( \ {H\choose r},{H\choose b},{H\choose b}-{r\choose \lambda}{H-r\choose b-\lambda},{H\choose r+b-2\lambda}{r+b-2\lambda\choose r-\lambda} \ )     $  PDA $\mathbf{P}'=(p'_{B,A})_{B\in \mathcal{F}, A\in \mathcal{K}}$
where
\begin{eqnarray}
\label{eq-rule1}
p_{B,A}=\left\{
\begin{array}{cc}
((A\cup B)-I,I)& \hbox{if} \ A\cap B=I\in \mathcal{I}\\
*&\hbox{Otherwise}
\end{array}
\right.
\end{eqnarray}
\begin{eqnarray}
\label{eq-rule2}
p'_{B,A}=\left\{
\begin{array}{cc}
\label{eq-rule2}
((A\cup B)-I,A-B)& \hbox{if} \ A\cap B=I\in \mathcal{I}\\
*&\hbox{Otherwise}
\end{array}
\right.
\end{eqnarray}
\end{construction}
\begin{example} When $H=4$, $r=2$, $b=2$, $\lambda=1$, we have
\begin{eqnarray*}
\mathcal{F}&=&\mathcal{K}=\{\{1,2\}, \{1,3\},  \{1,4\},\{2,3\},\{2,4\},\{3,4\}\}, \\
\mathcal{I}&=&\mathcal{I}'=\{\{1\},  \{2\},\{3\},\{4\}\}.
\end{eqnarray*} The following array $\mathbf{P}$ can be obtained by \eqref{eq-rule1}.
\begin{small}
\begin{eqnarray}
\label{eq-exam-1}
\begin{array}{c|ccccccc}
& 12& 13&  14& 23& 24& 34\\ \hline
12&*	    &	(23,1)	&	(24,1)	&	(13,2)	&	(14,2)	&	*	\\
13&(23,1)	&	*	    &	(34,1)	&	(12,3)	&	*	    &	(14,3)	\\
14&(24,1)	&	(34,1)	&	*	    &	*	    &	(12,4)	&	(13,4)	\\
23&(13,2)	&	(12,3)	&	*	    &	*	    &	(34,2)	&	(24,3)	\\
24&(14,2)	&	*	    &	(12,4)	&	(34,2)	&	*	    &	(23,4)	\\
34&*	    &	(14,3)	&	(13,4)	&	(24,3)	&	(23,4)	&	*	
\end{array}
\end{eqnarray}\end{small}
Here we represent each subset as a string for short. Replacing the entries $(23,1)$, $(24,1)$, $(13,2)$, $(14,2)$, $(34,1)$, $(12,3)$, $(14,3)$, $(12,4)$, $(13,4)$, $(34,2)$, $(24,3)$, $(23,4)$ of $\mathbf{P}$ in \eqref{eq-exam-1} by $0$, $1$, $\ldots$, $11$ respectively, the array in \eqref{eq-exam-Alg} is obtained.And the following array $\mathbf{P}'$ can be obtained by \eqref{eq-rule2}.
\end{example}

\begin{remark}
\label{remark-2}
The authors showed that the number of strong coloring of the bipartite graph generated in \cite{JA} is minimal. This implies that for the fixed placement strategy according to the bipartite graph in \cite{JA}, the PDA obtained by the bipartite graph in \cite{JA} has minimal value of $S$. The authors in \cite{YTCC} also showed that the scheme realized by this type of PDAs has significant advantages on subpacketization at the cost of increasing the transmission rate compared with MN PDA.
\end{remark}

\subsection{New schemes from the PDAs in \cite{YTCC}}
In the following we will show that this type of PDAs satisfies Theorem \ref{th-main} with $Z'>1$ for some parameters and the following result can be obtained.
\begin{theorem}
\label{th-main-1}
For any positive integers $H$, $r$, $b$, $\lambda$ satisfying $0<r,b<H$, $\lambda<\min\{r,b\}$ and $r+b\leq H$, there exists an $({H\choose b}-\sum_{i=0}^{\lambda-1}{r\choose i}{H-r\choose b-i})$-division $({H\choose r}$, $M,N)$ coded caching scheme with $\frac{M}{N}=1-\frac{{r\choose \lambda}{H-r\choose b-\lambda}}{{H\choose b}-\sum_{i=0}^{\lambda-1}{r\choose i}{H-r\choose b-i}}$ and transmission rate $R=\frac{{H\choose r+b-2\lambda}}{{H\choose b}-\sum_{i=0}^{\lambda-1}{r\choose i}{H-r\choose b-i}}{H-(r+b-2\lambda)\choose \lambda}$.
\end{theorem}
\begin{proof}
Let $\mathbf{P}$ be the PDA generated by  \eqref{eq-rule1}. From Theorem \ref{th-main}, we only need to count the number of useless stars in each column of $\mathbf{P}$. For any $A\in \mathcal{K}={[0,H)\choose r}$ and $B\in \mathcal{F}={[0,H)\choose b}$ satisfying $|A\cap B|<\lambda$, we have $p_{B,A}=*$ since $A\cap B \notin \mathcal{I}$. Moreover, if $|A\cap B|<\lambda$, the star $p_{B,A}$ is useless. Otherwise if $p_{B,A}=*$ is useful, which means that it occurs in a subarray $\mathbf{P}^{(C,I)}$ for some $C\in {[0,H)\choose r+b-2\lambda}$ and $I\in {[0,H)\choose \lambda}$. Then there must exist two subsets, say $A'\in \mathcal{K}\setminus\{A\}$ and $B'\in \mathcal{F}\setminus\{B\}$, such that
\begin{eqnarray*}
p_{B',A}&=&((A\cup B')-(A\cap B'),A\cap B')=p_{B,A'} \\
 &=&((A'\cup B)-(A'\cap B),A'\cap B)=(C,I).
\end{eqnarray*}
Then we have $A\cap B'=A'\cap B=I$. So we have $I\subseteq A\cap B$, which contradicts with $|A\cap B|<\lambda$ since $|I|=\lambda$.

Since for each $A\in \mathcal{K}$ there are exactly $\sum_{i=0}^{\lambda-1}{r\choose i}{H-r\choose b-i}$ subsets $B\in \mathcal{F}$ satisfying $|A\cap B|<\lambda$. So each column has $Z'=\sum_{i=0}^{\lambda-1}{r\choose i}{H-r\choose b-i}$ useless stars. Based on the first PDA $\mathbf{P}$ in Construction \ref{constrct1} and using Theorem \ref{th-main}, our statement holds.
\end{proof}
\begin{theorem}
\label{th-main-2}
For any positive integers $H$, $r$, $b$, $\lambda$ satisfying $0<r,b<H$, $\lambda<\min\{r,b\}$ and $r+b\leq H+\lambda$, there exists an $({H\choose b}-\sum_{i=\lambda+1}^{r}{r\choose i}{H-r\choose b-i})$-division $({H\choose r}$, $M,N)$ coded caching scheme with $\frac{M}{N}=1-\frac{{r\choose \lambda}{H-r\choose b-\lambda}}{{H\choose b}-\sum_{i=\lambda+1}^{r}{r\choose i}{H-r\choose b-i}}$ and transmission rate $R=\frac{{H\choose r+b-2\lambda}}{{H\choose b}-\sum_{i=\lambda+1}^{r}{r\choose i}{H-r\choose b-i}}{r+b-2\lambda\choose r-\lambda }$.
\end{theorem}
\begin{proof}
Let $\mathbf{P'}$ be the PDA generated by  \eqref{eq-rule2}. For any $A\in \mathcal{K}={[0,H)\choose r}$ and $B\in \mathcal{F}={[0,H)\choose b}$ with $|A\cap B|>\lambda$, we have  $p'_{B,A}=*$ since $A\cap B \notin \mathcal{I}$. We claim that such $p'_{B,A}=*$ is useless. Otherwise if the star $p'_{B,A}$ is useful. Then there must exist two subsets, say $A'\in \mathcal{K}\setminus\{A\}$ and $B'\in \mathcal{F}\setminus\{B\}$ satisfying $|A\cap B'|=\lambda$ and $|A'\cap B|=\lambda$, such that
\begin{eqnarray*}
p'_{B',A}&=&((A\cup B')-(A\cap B'),A- B') =p'_{B,A'}\\
  &=& ((A'\cup B) - (A' \cap B ),  A'- B).
\end{eqnarray*}

Then we have $A- B'=A'- B$. So we have $(A- B')\cap B=(A'- B)\cap B=\emptyset$. On the other hand, since $|A\cap B|>\lambda$ and $|A\cap B'|=\lambda$, there exists some $x\in [0,H)$ satisfying $x\in A\cap B$ and $x\notin A\cap B'$. Consequently we have $x\in A$, $x\in B$ and $x\notin B'$. So we have $x\in (A-B')\cap B$, which contradicts with $(A- B')\cap B=\emptyset$.

Since for each $A\in \mathcal{K}$ there are exactly $\sum_{i=\lambda+1}^{r}{r\choose i}{H-r\choose b-i}$ subsets $B\in \mathcal{F}$ satisfying $|A\cap B|>\lambda$. So each column of $\mathbf{P'}$ has $Z'=\sum_{i=\lambda+1}^{r}{r\choose i}{H-r\choose b-i}$ useless stars. Based on the second PDA $\mathbf{P'}$ in Construction \ref{constrct1} and using Theorem \ref{th-main}, our statement holds.
\end{proof}

\subsection{Performance analyses}
\label{sec-Panalyses}
In this section, we assume the parameters $H$ and $r$ are fixed.  The scheme in Table \ref{tab-known} is denoted by the original scheme, and the scheme with smaller transmission rate among the schemes in Table \ref{tab-main} (i.e. the schemes by Theorem \ref{th-main-1} and Theorem \ref{th-main-2})  is called the new scheme. Since it is hard to propose a theoretic comparison between the original scheme and the new scheme, we take $H=10$ and $r=5$ and compare their transmission rates and the memory ratios,  see Figure \ref{Fig_R}.
From Figure \ref{Fig_R}, it's easy to see that the transmission rate of the original scheme is much larger than that of the new scheme when the memory ratio is small.
\begin{figure}[!htbp]
  \centering
  \includegraphics[scale=0.3]{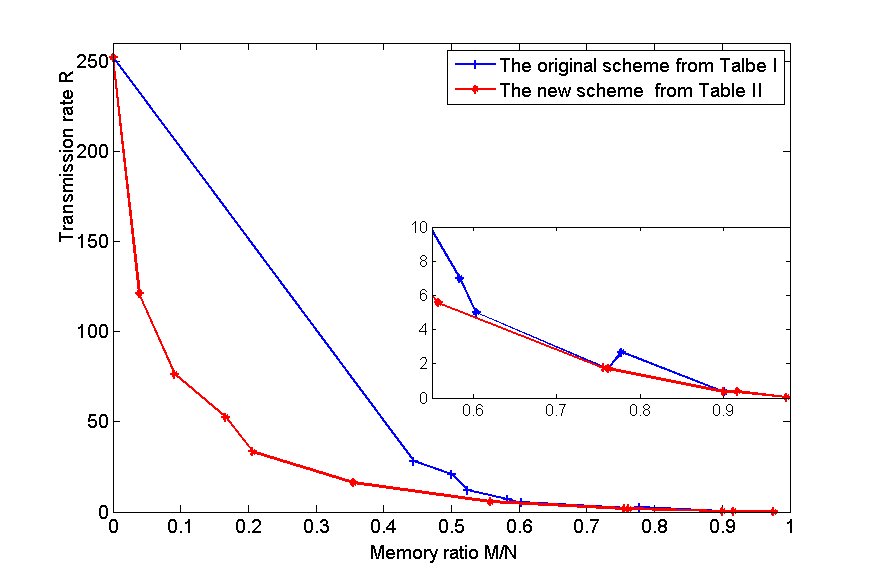}\\
  \caption{The transmission rate of the schemes from Table \ref{tab-known} and Table \ref{tab-main} where $H=10$ and $r=5$.}\label{Fig_R}
\end{figure}

\section{Conclusion}
\label{conclusion}
In this paper, we  studied coded caching schemes with low subpacketizations and small memory ratios. After we  observed that there are some packets cached by users do not generate multicasting opportunities in the delivery phase,
we modified the uncoded placement of the scheme realized by some appropriate PDAs to coded placement so that each packet cached by users can generate multicasting opportunities. Finally we  used the PDAs in \cite{YTCC} in our approach to construct two  new schemes which  has smaller memory ratio and transmission rate than that of the original scheme in \cite{YTCC}.

\ifCLASSOPTIONcaptionsoff
  \newpage
\fi



%

\end{document}